\documentclass[12pt]{article}

\usepackage{a4}
\usepackage{amsthm}
\usepackage{amsfonts}
\usepackage{amssymb}
\usepackage{amsmath} 
\usepackage{array}
\usepackage{stmaryrd}
\usepackage{graphicx}
\usepackage{hyperref}

\newtheorem{theorem}{Theorem}
\newtheorem{lemma}[theorem]{Lemma}

\newtheorem{observation}[theorem]{Observation}

\def\AA{{\cal A}}
\def\BB{{\cal B}}
\def\ff{f_{\bullet}}
\def\ttq{\mathbf{t}^{q}}

\newcommand{\NN}{\mathbb{N}}
\newcommand{\PP}{\mathcal{P}}

\newcommand{\perm}{\mathfrak{S}}
\newcommand{\permutons}{\mathfrak{P}}
\newcommand{\Jac}{\operatorname{Jac}}
\newcommand{\Mon}{\operatorname{Mon}}
\newcommand{\Ind}{\operatorname{Occ}}
\newcommand{\Hom}{\operatorname{Hom}}
\newcommand{\mon}{\operatorname{mon}}
\newcommand{\simplif}{\mathcal{R}}

\begin{document}

\title{Densities in large permutations and parameter testing}

\author{
Roman Glebov\thanks{School of Computer
Science and Engineering, Hebrew University, Jerusalem, 9190401. E-mail: {\tt roman.l.glebov@gmail.com}.}\and 
Carlos Hoppen\thanks{Instituto de Matem\'atica, UFRGS -- Avenida Bento Gon\c{c}alves, 9500, 91509-900, Porto Alegre, RS, Brazil. E-mail: {\tt choppen@ufrgs.br}.}\and 
Tereza Klimo\v sov\' a\thanks{Department of Applied Mathematics, Faculty of Mathematics and Physics, Charles University,  
Malostransk\'e n\'am\v{e}st\'i 25, 118 00 Praha 1, Czech Republic. E-mail: {\tt tereza@kam.mff.cuni.cz}.}\and Yoshiharu Kohayakawa\thanks{Instituto de Matem\'atica e Estat\'\i stica, USP -- Rua do Mat\~ao
  1010, 05508--090 S\~ao Paulo, SP, Brazil. E-mail: {\tt yoshi@ime.usp.br}.}\and  
Daniel Kr\' al'\thanks{Mathematics Institute, DIMAP and Department of Computer Science, University of Warwick, Coventry CV4 7AL, UK. E-mail: {\tt d.kral@warwick.ac.uk}.}\and 
Hong Liu\thanks{Mathematics Institute and DIMAP, University of Warwick, Coventry CV4 7AL, UK. E-mail: {\tt h.liu.9@warwick.ac.uk}.}}

\maketitle

\begin{abstract}
A classical theorem of Erd\H os, Lov\'asz and Spencer asserts that
the densities of connected subgraphs in large graphs are independent.
We prove an analogue of this theorem for permutations and
we then apply the methods used in the proof to give an example of a finitely
approximable permutation parameter that is not finitely forcible.
The latter answers a question posed by two of the authors and Moreira and Sampaio.
\end{abstract}

\section{Introduction}
Computer science applications that involve large networks form one of the main motivations
to develop methods for the analysis of large graphs.
The theory of graph limits, which emerged in a series of papers
by Borgs, Chayes, Lov\'asz, S\'os, Szegedy and Vesztergombi~\cite{bib-borgs06+,bib-borgs+,bib-borgs08+,bib-lovasz06+},
gives analytic tools to cope with problems related to large graphs.
It also provides an analytic view of many standard concepts,
e.g. the regularity method~\cite{bib-lovasz07} or property testing algorithms~\cite{bib-hoppen-test,bib-lovasz10+}.
In this paper, we focus on another type of discrete objects, permutations, and
we give permutation counterparts of some of classical results on large graphs.
It is worth noting that not all results on large graphs have permutation analogues and vice versa
as demonstrated, for example, by the finite forcibility of graphons and permutons~\cite{bib-ff-perm} (vaguely speaking, finite forcibility means that a global structure is determined by finitely many substructure densities).

Both our main results are related to the dependence of possible densities of (small) substructures.
In the case of graphs, Erd\H os, Lov\'asz and Spencer~\cite{bib-ELS} considered three notions of substructure
densities: the subgraph density, the induced subgraph density and the homomorphism density.
They showed that these types of densities in a large graph are strongly related and that the densities of connected graphs are independent in the sense that none of the densities can be expressed as a function of the others.
The result has a natural formulation in the language of graph limits, which are called graphons:
the body of possible densities of any $k$ connected graphs in graphons,
which is a subset of $[0,1]^k$, has a non-empty interior (in particular, it is full dimensional).

Our first result asserts that the analogous statement is also true for permutations.
As in the case of graphs, it is natural to cast our result in terms of permutation limits, called permutons.
The theory of permutation limits was initiated in~\cite{bib-hoppen-lim, bib-hoppen-lim2} (also see~\cite{bib-presutti}) and successfully applied e.g.~in~\cite{bib-hoppen-test,bib-kral-pikhurko}.
To state our first result, we use the notion of a {\em indecomposable} permutation, which is an analogue of graph connectivity in the sense that an indecomposable permutation cannot be split into {\em independent parts}. Let $T^q$ be the body of possible densities of indecomposable permutations of order at most $q$ in a permuton (a precise definition and further details can be found in Section~\ref{sec:basics}). Our first result says that $T^q$ has a non-empty interior for every $q$. In particular, it contains $B(\mathbf{w},\varepsilon)$, for some vector $\mathbf{w}$ and some $\varepsilon>0$, where $B(\mathbf{w},\varepsilon)$ denotes the ball of radius $\varepsilon$ around $\mathbf{w}$.

\begin{theorem}\label{thm:ball}
For every integer $q\geq 2$, there exist a vector $\mathbf{w}\in T^{q}$ and $\varepsilon>0$ such that $B(\mathbf{w},\varepsilon)\subseteq T^{q}$.
\end{theorem}

Our second result is related to algorithms for large permutations.
Such algorithms are counterparts of extensively studied graph property testing, see e.g. ~\cite{bib-alon09+, bib-alon08mono+, bib-goldreich96+, bib-goldreich03+, bib-rodl85+}.
In the case of permutations, two of the authors and Moreira and Sampaio~\cite{bib-hoppen-test-SODA, bib-hoppen-test}
established that every hereditary permutation property is testable with respect to the rectangular distance and two of the other authors~\cite{bib-soda} strengthened the result to testing with respect to Kendall's tau distance.
In addition to property testing,
a related notion of parameter testing was also considered in~\cite{bib-hoppen-test} where  testable bounded permutation parameters were characterized.

However, the interplay between testing and the finite forcibility of permutation parameters
was not fully understood in~\cite{bib-hoppen-test}.
In particular, the authors asked~\cite[Question 5.5]{bib-hoppen-test} 
whether there exists a testable bounded permutation parameter that is not finitely forcible.
Our second result gives a positive answer to this question.

\begin{theorem}\label{thm:main}
There exists a bounded permutation parameter $f$ that is finitely approximable but not finitely forcible.
\end{theorem}

Informally speaking, we utilize the methods used in the proof of Theorem~\ref{thm:ball} to construct a permutation parameter that oscillates on indecomposable permutations, with bounded amplitude, so that the parameter testable though it fails to be finitely forcible.

\section{Preliminaries}

In this section, we introduce the notions used throughout the paper.
Most of our notions are standard but we include all of them
for the convenience of the reader.

\subsection{Permutations}\label{sec:basics}

A {\em permutation of order $n$} is a bijective mapping from $[n]$ to $[n]$, where $[n]$ denotes the set $\{1,\ldots,n\}$. 
The order of a permutation $\sigma$ is denoted by $|\sigma|$. We say a permutation is {\em non-trivial} if it has order greater than $1$. We denote by $S_n$ the set of all permutations of order $n$ and let  $\perm=\bigcup_{n\in\mathbb{N}}S_n$.
An {\em inversion} of a permutation $\sigma$ is a pair $(i,j)$, $i,j\in[|\sigma|]$, such that $i<j$ and $\sigma(i)>\sigma(j)$. An {\em interval} $I$ in $[m]$ is a set of integers of the form $\{k \mid a\leq k\leq b\}$ for some $a,b\in [m]$. An interval $I$ is {\em proper} if $a<b$ and $I\neq[m]$.

We say that a permutation $\sigma$ of order $n$ is {\em indecomposable} if there is no $1\leq m<n$ such that $\sigma([m])=[m]$. Note that 
\begin{eqnarray}\label{eq:conn}
\Pr_{\sigma\in S_n}(\sigma \mbox{ is not indecomposable})&\leq&\frac{\sum_{m=1}^{n-1} m!(n-m)!}{n!}=\sum_{m=1}^{n-1}\binom{n}{m}^{-1}\\
&\leq& \frac{2}{n}+\sum_{m=2}^{n-2}\binom{n}{m}^{-1}\leq \frac{2}{n}+(n-3)\frac{2}{n(n-1)}.\nonumber
\end{eqnarray}
Thus, $\lim_{n\rightarrow \infty}\Pr_{\sigma\in S_n}(\sigma \mbox{ is indecomposable})=1$.

We say that a permutation $\sigma$ is {\em simple} if it does not map any proper interval onto an interval. For example the permutation $(\sigma(1),\ldots,\sigma(4))=(2,4,1,3)$ is simple. 

Albert, Atkinson and Klazar~\cite{bib-klazar} showed that a random permutation is simple with a probability bounded away from zero. Specifically, they proved the following.
\begin{equation}\label{eq:simple}
\lim_{n\rightarrow \infty}\mathbb{P}_{\sigma\in S_n}(\sigma \mbox{ is simple})=e^{-2}.
\end{equation}

Let $\pi$ be a permutation of order $k$ and $\sigma$ a permutation of order $n$.
We introduce three ways in which $\pi$ can appear in $\sigma$: as a subpermutation, through a monomorphism and through a homomorphism. We say that $\pi$ is a {\em subpermutation of $\sigma$} if there exists a strictly increasing function $f:[k]\rightarrow[n]$, such that $\pi(i)>\pi(j)$ if and only if $\sigma(f(i))>\sigma(f(j))$ for every $i,j\in [k]$. We then say that $f([k])$ {\em induces} a subpermutation $\pi$ in $\sigma$. Let  $\Ind(\pi,\sigma)$ be the set of all such functions $f$ from $[k]$ into $[n]$ and let $\Lambda(\pi,\sigma)=|\Ind(\pi,\sigma)|$. The {\em density} of $\pi$ in $\sigma$ is defined as
\begin{equation*}
t(\pi,\sigma)=  
\left\{ \begin{array}{ll}
\Lambda(\pi,\sigma) \binom{n}{k}^{-1} & \mbox{ if }k\leq n\mbox{ and }\\
0 & \mbox{ otherwise.}\\
\end{array}\right.
\end{equation*}
A non-decreasing function $f:[k]\rightarrow[n]$ is a {\em homomorphism} of $\pi$ to $\sigma$ if $\sigma(f(i))>\sigma(f(j))$ for every $i,j\in [k]$ such that  $i<j$ and $\pi(i)>\pi(j)$, that is, $f$ preserves inversions.
A {\em monomorphism} is a homomorphism that is injective.

Let $\Hom(\pi,\sigma)$ and $\Mon(\pi,\sigma)$ be the sets of homomorphisms and monomorphisms of $\pi$ to $\sigma$, respectively, and let $\Lambda_{\hom}(\pi,\sigma)$ and $\Lambda_{\mon}(\pi,\sigma)$ denote the sizes of the respective sets. Note that $\Ind(\pi,\sigma)\subseteq \Mon(\pi,\sigma)\subseteq \Hom(\pi,\sigma)$. The  {\em homomorphism density} $t_{\hom}$ and {\em monomorphism density} $t_{\mon}$ are defined as follows:
\begin{align*}
t_{\mon}(\pi,\sigma)&=  
\left\{ \begin{array}{ll}
\Lambda_{\mon}(\pi,\sigma)\binom{n}{k}^{-1} & \mbox{ if }k\leq n\mbox{ and }\\
0 & \mbox{ otherwise,}\\
\end{array}\right.\\
t_{\hom}(\pi,\sigma)&=\Lambda_{\hom}(\pi,\sigma)\binom{n+k-1}{k}^{-1}.
\end{align*}
The three densities that we have just introduced are analogues of the induced subgraph density, homomorphism density and subgraph density  for graphs studied in~\cite{bib-ELS}.

Let $q$ be an integer and let $\{\tau_1,\ldots, \tau_r\}$ be the set of all non-trivial indecomposable permutations of order at most $q$. We consider the following three vectors
\begin{align*}
\ttq(\sigma)&=(t(\tau_1,\sigma),\ldots,t(\tau_r,\sigma)),\\
\ttq_{\mon}(\sigma)&=(t_{\mon}(\tau_1,\sigma),\ldots,t_{\mon}(\tau_r,\sigma))\mbox{, and}\\
\ttq_{\hom}(\sigma)&=(t_{\hom}(\tau_1,\sigma),\ldots,t_{\hom}(\tau_r,\sigma)).
\end{align*}
Our aim is to understand possible densities of subpermutations in large permutations. This leads to the following definitions, which reflect the possible asymptotic densities of the indecomposable permutations of order at most $q$ in permutations:
\begin{align*}
T^{q}&=\{\mathbf{v}\in\mathbb{R}^r\mid \exists (\sigma_n)_{n=1}^{\infty} \mbox{ such that } \ttq(\sigma_n)\rightarrow \mathbf{v}\mbox{ and }|\sigma_n|\rightarrow \infty\},\\
T^{q}_{\mon}&=\{\mathbf{v}\in\mathbb{R}^r\mid \exists (\sigma_n)_{n=1}^{\infty} \mbox{ such that } \ttq_{\mon}(\sigma_n)\rightarrow \mathbf{v}\mbox{ and }|\sigma_n|\rightarrow \infty\}\mbox{, and}\\
T^{q}_{\hom}&=\{\mathbf{v}\in\mathbb{R}^r\mid \exists (\sigma_n)_{n=1}^{\infty} \mbox{ such that } \ttq_{\hom}(\sigma_n)\rightarrow \mathbf{v}\mbox{ and }|\sigma_n|\rightarrow \infty\}.
\end{align*}

In Section~\ref{sec:lim}, we will see that $T^{q}$ and $T^{q}_{\mon}$ have another, simpler definition in language of permutons. Now we give three observations on how the sets $T^{q}$, $T^{q}_{\mon}$ and $T^{q}_{\hom}$ relate to each other.

\begin{observation}\label{obs:same}
The sets $T^{q}_{\mon}$ and $T^{q}_{\hom}$ are equal for every $q\in \mathbb{N}$.
\end{observation}

\begin{proof}
Observe that for every fixed integer $k$, 
\[\Lambda_{\hom}(\tau,\sigma)-\Lambda_{\mon}(\tau,\sigma)\le \binom{k}{2}{n}^{k-1}=O(n^{k-1}),\]
for every $\sigma$ of order $n$ and $\tau$ of order $k$.

Hence, for every permutation $\tau$ and every real $\varepsilon>0$ there exists $n_0$ such that $|t_{\mon}(\tau,\sigma)-t_{\hom}(\tau,\sigma)|<\varepsilon$ for every permutation $\sigma$ with $|\sigma|>n_0$. The statement now follows.
\end{proof}

In view of Observation~{\ref{obs:same}, we will discuss only $T^{q}_{\mon}$ in the rest of the paper.

\begin{observation}\label{obs:closed}
For every $q\in\mathbb{N}$, the set $T^{q}_{\mon}$ is closed.
\end{observation}

\begin{proof}
Consider a convergent sequence $(\mathbf{w}_{n})_{n\in\mathbb{N}}\subseteq T^{q}_{\mon}$ and let $\mathbf{w}=\lim_{n\rightarrow\infty}\mathbf{w}_n$. For each $n$, choose $\sigma_{n}$ such that $\|\ttq_{\mon}(\sigma_n)-\mathbf{w}_n\|\leq 1/n$. Observe that $\ttq_{\mon}(\sigma_n)$ converges to $\mathbf{w}$.
\end{proof}

\begin{observation}\label{obs:lin-transf}
The set $T^{q}$ is a non-singular linear transformation of $T^{q}_{\mon}$  for every $q\in \mathbb{N}$.
\end{observation}

\begin{proof}
Note that $\Lambda_{\mon}(\pi,\sigma)=\sum_{\pi'\in \mathcal{P}}\Lambda(\pi',\sigma)$, where $\mathcal{P}$ is a set of permutations $\pi'$ of the same order as $\pi$ such that the identity mapping is a monomorphism from $\pi$ to $\pi'$. Consequently, $t_{\mon}(\pi,\sigma)=\sum_{\pi'\in \mathcal{P}}t(\pi',\sigma)$. This gives that $T^q_{\mon}$ is a linear transformation of $T^q$. 
Observe that if we order $\tau_1,\dots, \tau_r$ by the number of inversions, the coefficient matrix of the induced linear mapping is upper triangular with diagonal entries equal to $1$. We conclude that the linear transformation of
$T^q$ is non-singular.
\end{proof}

\subsection{Permutation limits}\label{sec:lim}

In this subsection, we survey the theory of permutation limits,
which was introduced in~\cite{bib-hoppen-lim,bib-hoppen-lim2} (a similar representation was used in~\cite{bib-presutti}).
We follow the terminology used in~\cite{bib-kral-pikhurko}. 
An infinite sequence $(\sigma_i)_{i\in\mathbb{N}}$ of permutations with $|\sigma_i|\to\infty$ is {\em convergent} if $t(\tau,\sigma_i)$ converges for every permutation $\tau\in\perm$. Observe that every sequence of permutations has a convergent subsequence. A convergent sequence can be associated with an analytic limit object, a {\em permuton}. A {\em permuton} is a probability measure $\Phi$ on the $\sigma$-algebra of Borel sets of the unit square $[0,1]^2$ such that $\Phi$ has {\em uniform marginals}, i.e., $\Phi\left(\left[\alpha,\beta\right]\times [0,1]\right)=\Phi\left([0,1]\times[\alpha,\beta]\right)=\beta-\alpha$ for every $0\le\alpha\le\beta\le 1$. We denote the set of all permutons by $\permutons$.

Given a permuton $\Phi$, a {\em $\Phi$-random permutation of order $n$} is a permutation $\sigma_{\Phi,n}$ obtained in the following way. Sample $n$ points $(x_1,y_1),\ldots,(x_n,y_n)$ in $[0,1]^2$ at random with the distribution given by $\Phi$. Note that the values of $x_i$ are pairwise distinct with probability one and the same holds for the values of $y_i$. Let $i_1,\ldots,i_n\in [n]$ be such that $x_{i_1}<x_{i_2}<\cdots<x_{i_n}$. Then the permutation $\sigma_{\Phi,n}$ is the unique bijective mapping from $[n]$ to $[n]$ satisfying that $\sigma_{\Phi,n}(j)<\sigma_{\Phi,n}(j')$ if and only if $y_{i_j}<y_{i_{j'}}$ for every $j,j'\in[n]$. Informally speaking, the values  $x_i$ determine the ordering of the points and the relative order of the values $y_i$ determines the relative order of the elements of the permutation.

If $\Phi$ is a permuton and $\sigma$ is a permutation of order $n$, then $t(\sigma,\Phi)$ is the probability that a $\Phi$-random permutation of order $n$ is $\sigma$. We say that a  permuton $\Phi$ is a {\em limit} of a convergent sequence of permutations $(\sigma_i)_{i\in\mathbb{N}}$ if \[\lim\limits_{i\rightarrow\infty}t(\tau,\sigma_i)=t(\tau,\Phi)\] for every $\tau\in \perm$. Every convergent sequence of permutations has a limit and the permuton representing the limit of a convergent sequence
of permutations is unique.

We now give some examples of the notions we have just defined (the corresponding permutons are depicted
in Figure~\ref{fig-permuton}).
Let us consider a sequence $\left(\pi^1_i\right)_{i\in \NN}$ such that $\pi^1_i$ is the identity permutation of order $i$,
i.e., $\pi^1_i(k)=k$ for $k\in [i]$. This sequence is convergent and its limit is the permuton $I$ with support $\left\{(x,x), x\in [0,1]\right\}$ and measure uniformly distributed on its support.
Similarly, the limit of a sequence $\left(\pi^2_i\right)_{i\in \NN}$, where $\pi^2_i$ is the permutation of order $i$
defined as $\pi^2_i(k)=i+1-k$ for $k\in [i]$, is the permuton with support $\left\{(x,1-x), x\in [0,1]\right\}$ and measure uniformly distributed on its support.

\begin{figure}
\begin{center}
\includegraphics[scale=1]{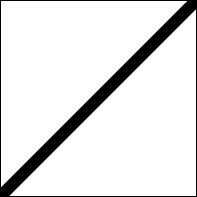} \hskip 10mm
\includegraphics[scale=1]{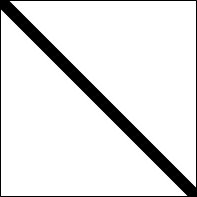} \hskip 10mm
\includegraphics[scale=1]{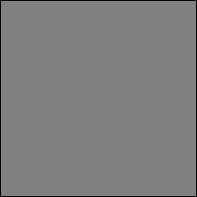}
\end{center}
\caption[The limits of permutation sequences.]{The limits of sequences $\left(\pi^1_i\right)_{i\in \NN}$, $\left(\pi^2_i\right)_{i\in \NN}$, $\left(\pi^3_i\right)_{i\in \NN}$.}\label{fig-permuton}
\end{figure}

Next, consider a sequence $(\pi^3_i)_{i\in \NN}$ such that $\pi^3_i$ is a uniformly random permutation of order $i$.
This sequence is convergent with probability one and its limit is the uniform probability measure on $[0,1]^2$ with probability one. 

Similarly to the subpermutation density, we can define the monomorphism density of a permutation $\tau$ in a permuton $\Phi$ as the probability that the identity mapping to a random $\Phi$-permutation is a monomorphism of $\tau$. Since we view permutons as representing large permutations, if we defined
homomorphism densities in a natural way, they would coincide with
monomorphism densities. So, we restrict our study to subpermutation
densities and monomorphism densities in permutons. By analogy to the finite case, we define the vectors
\begin{align*}
\ttq(\Phi)&=(t(\tau_1,\Phi),\ldots,t(\tau_r,\Phi))\mbox{ and }\\
\ttq_{\mon}(\Phi)&=(t_{\mon}(\tau_1,\Phi),\ldots,t_{\mon}(\tau_r,\Phi)),
\end{align*}
where $q\in \mathbb{N}$ and $\{\tau_1,\ldots, \tau_r\}$ is the set of all non-trivial indecomposable permutations of order at most $q$.

If $\Phi$ is a permuton and $\sigma_i$ is a $\Phi$-random permutation of order $i$, then the sequence $(\sigma_i)_{i\in\mathbb{N}}$ is convergent with probability one and $\Phi$ is its limit. In particular, this means that for every finite set of permutations $\mathcal{P}$ and every $\varepsilon>0$, there exists a permutation $\varphi$ such that $|t(\pi,\Phi)-t(\pi,\varphi)|<\varepsilon$ for every $\pi\in \mathcal{P}$. This yields an alternative description of $T^{q}$ as the set $\{\ttq(\Phi)\mid \Phi\in \permutons\}$.
Similarly, $T^{q}_{\mon}=\{\ttq_{\mon}(\Phi)\mid \Phi\in \permutons\}$.

\subsection{Permuton constructions}
In this section we introduce constructions of {\em step-up} permutons and a {\em direct sum} of permutons, which we use in Section~\ref{sec:Tq}, and we derive formulas for densities of indecomposable subpermutations in the constructed permutons.

The step-up permutons are permutons with simple structure corresponding to a weighted permutation. They are defined as follows. Let $\sigma$ be a permutation of order $n$ and let $\mathbf{v}=(v_1,\ldots, v_{n})\in \mathbb{R}_{+}^{n}$ be such that $\sum_{i\in[n]}v_i\leq 1$, where $\mathbb{R}_{+}$ is the set of positive reals. The {\em step-up permuton} of $\sigma$ and $\mathbf{v}$ is the permuton $\Phi_{\sigma}^{\mathbf{v}}$ such that the support of the measure $\Phi_{\sigma}^{\mathbf{v}}$ is formed by the segments between the points 
$(\sum_{j< i} v_j,\sum_{\sigma(j)<\sigma(i)} v_j)$ and $(\sum_{j\leq i} v_j,\sum_{\sigma(j)\leq \sigma(i)} v_j)$ for $i\in[n]$ and the segment between the points $(\sum_{j=1}^n v_j,\sum_{j=1}^n v_j)$ and $(1,1)$. Note that this uniquely determines the permuton $\Phi_{\sigma}^{\mathbf{v}}$ because it must have uniform marginals. See Figure~\ref{fig:step-up} for an example.

\begin{figure}
\begin{center}
\includegraphics[scale=0.5]{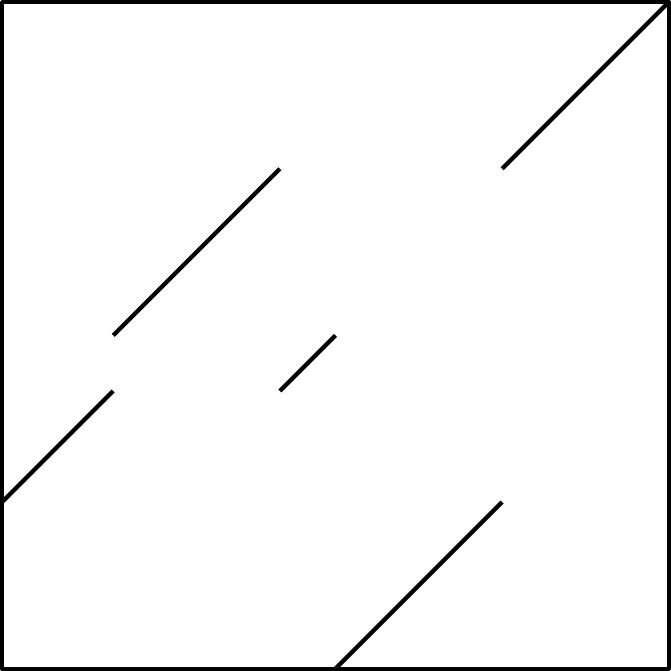}
\end{center}
\caption{The permuton $\Phi_{\sigma}^{\mathbf{v}}$ for $\sigma= (2,4,3,1)$ and $\mathbf{v}=(1/6,1/4,1/12,1/4)$.}\label{fig:step-up}
\end{figure}

\begin{figure}
\begin{center}
\includegraphics[scale=1]{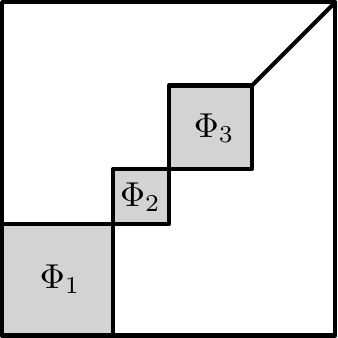}
\end{center}
\caption{The permuton $(1/3,\Phi_1)\oplus(1/6,\Phi_2)\oplus(1/4,\Phi_3)$.}\label{fig:oplus}
\end{figure}

We now define the direct sum of permutons with weights. For $k\in \NN$, a sequence of permutons $(\Phi_{i})_{i\in[k]}$ and $(p_{i})_{i\in[k]}\in \mathbb{R}_{+}^k$ such that $\sum_{i\in[k]} p_i\leq 1$, the {\em direct sum of permutons} $\Phi_i$ with weights $p_i$ is denoted by $\Phi=\bigoplus_{i\in[k]} (p_i,\Phi_{i})$ and is defined as follows:

\[\Phi(S)=\sum_{i=1}^{k+1} p_i\Phi_i(\theta_i(S\cap C_i))\] for every Borel set $S$, where $\Phi_{k+1}=I$ (the first permuton in Figure~\ref{fig-permuton}), $p_{k+1}=1-\sum_{i=1}^kp_i$,
\begin{equation}\label{eq:c}
C_i=\left[\sum_{j=1}^{i-1}p_j,\sum_{j=1}^{i}p_j\right]^2
\end{equation} 
and
$\theta_i$ is a map from $C_i$ to $[0,1]^2$  defined as
\[\theta_i((x,y))=\left(\frac{x-\sum_{j=1}^{i-1}p_j}{p_i}, \frac{y-\sum_{j=1}^{i-1}p_j}{p_i}\right)\] for every $i\in [k+1]$. See Figure~\ref{fig:oplus} for an example.

For a permutation $\tau$ of order $k$, we call an ordered partition $\PP=(P_1, \ldots, P_{\ell})$ of~$[k]$ a {\em $\tau$-compressive partition} if 
\begin{itemize}
\item $P_i$ is an interval for every $i\in [\ell]$,
\item $a<b$ for every $a\in P_i$ and $b\in P_j$ with $i<j$, and
\item for every $i\in[\ell]$, there exists an integer $c_i$, such that $\tau(a)=a+c_i$ for every $a\in P_i$. (In particular, $\tau(P_i)$ is an interval for every $i\in [\ell]$.)
\end{itemize}

We denote the set of all $\tau$-compressive partitions by $\simplif(\tau)$. Note that for every permutation $\tau$, there exist at least one $\tau$-compressive partition: the partition into singletons.

For a permutation $\tau$ of order $k$ and a $\tau$-compressive partition $\PP=(P_1, \ldots, P_{\ell})$, let $\tau/\PP$ be a subpermutation of $\tau$ of order $\ell$ induced by $\{a_1, \ldots,  a_{\ell}\}$ where $a_i\in P_i$ for every $i\in[\ell]$. 
Note that  $\tau/\PP$ is unique, in particular, it is independent of the choice of the elements $a_i$.

In other words, the permutation $\tau/\PP$ is a permutation that can be obtained from $\tau$ by shrinking each interval $P_i$ and its image into single points, without changing the relative order of the elements of the permutation. For instance, $\PP=(\{1,2\},\{3\},\{4,5\})$ is a $(4,5,1,2,3)$-compressive partition, with $(4,5,1,2,3)/\PP=(3,1,2)$.

\begin{observation}\label{obs:t-step-up}
Let $\tau$ be a non-trivial indecomposable permutation of order $k$, $\sigma$ a permutation of order $n\geq k$ and let $\mathbf{p}=(p_1,\ldots, p_n)\in \mathbb{R}_{+}^n$ be such that $\sum_{i\in [n]} p_i\leq 1$. 
It follows that 

\[ t(\tau,\Phi_{\sigma}^{\mathbf{p}})=k!\sum_{\PP\in\simplif(\tau)}
\sum_{\psi\in \Ind(\tau/\PP, \sigma)}\prod_{i=1}^{|\PP|} \frac{p^{|P_i|}_{\psi(i)}}{|P_i|!}\;.\]
\end{observation}

\begin{proof}
 Consider $k$ distinct points in the support of $\Phi_{\sigma}^{\mathbf{p}}$ and label them $(x_i,y_i)$, $i\in [k]$, in such a way that $x_i<x_j$ if $i<j$ for every $i,j\in [k]$.
Let $i_1<i_2< \cdots <i_{k'}$ be the indices of the segments of the support (numbered from left to right) that contain at least one of the points and let $\PP=\{P_1,\ldots, P_{k'}\}$ be a partition of~$[k]$ such that $i\in P_j$ if the point $(x_i,y_i)$ lies on the $i_j$-th segment. 

Assume that the points yield the permutation $\tau$. Then, $\PP$ is a $\tau$-compressive partition and the subpermutation $\sigma'$ of $\sigma$ induced by $\{i_1,\ldots, i_{k'}\}$ is $\tau/\PP$. (Note that since $\tau$ is irreducible, none of the points lies on the $(k+1)$-st segment of the support of $\Phi_{\sigma}^{\mathbf{p}}$.)
The converse is also true; fix $k'$ segments with indices $1\leq i_1<\cdots<i_{k'}\leq n$ and a $\tau$-compressive partition $\mathcal{P}=\{P_1,\ldots,P_{k'}\}$ such that the subpermutation $\sigma'$ of $\sigma$ induced by $\{i_1,\ldots,i_{k'}\}$ is $\tau/\mathcal{P}$. Then any choice of points $(x_1,y_1),\ldots,(x_k,y_k)$ where each $(x_t,y_t)$ lies on the segment $i_j$ such that $t \in P_{i_j}$ yields the permutation $\tau$. 

Note that $k$ random points chosen based on the distribution $\Phi_{\sigma}^{\mathbf{p}}$ are distinct and lie in the support of $\Phi_{\sigma}^{\mathbf{p}}$ with probability one. The probability that they correspond to a given $\tau$-compressive partition $\PP$ and $\psi\in \Ind(\tau/\PP,\sigma)$ is $k!\prod_{i=1}^{|\PP|}\left( p^{|P_i|}_{\psi(i)}/|P_i|!\right)$. Since these events are disjoint for different pairs $(\mathcal{P},\psi)$, the result follows. 
\end{proof}

\begin{observation}\label{obs:t-composed}
Let $\tau$ be a non-trivial indecomposable permutation of order $k$ and let $m$ be a positive integer. Let $\Phi_1,\ldots, \Phi_m$ be permutons and let $\mathbf{x}=(x_1,\ldots,x_m)\in \mathbb{R}_{+}^m$ be such that $\sum_{i\in [m]} x_i\leq 1$. The permuton $\Phi^{\mathbf{x}}=\bigoplus_{i\in[m]}(x_i, \Phi_i)$ satisfies 
\[t(\tau,\Phi^{\mathbf{x}})=\sum_{i=1}^m x_i^{k} t(\tau,\Phi_i).\]
\end{observation}

\begin{proof}
Observe that, if $k$ random points chosen based on the distribution $\Phi^{\mathbf{x}}$ yield an indecomposable permutation, then all the points lie in the same $C_i$, for some $i\in [m]$ (where $C_i$ is given by~\ref{eq:c} in the definition of the direct sum of permutons). The probability that all the points are in $C_i$ is $x_i^k$ since $\Phi^{\mathbf{x}}(C_i)=x_i$.
Conditioned on this event, the probability that the points yield $\tau$ is $t(\tau,\Phi_i)$. The result follows.
\end{proof}

Analogues of Observations~\ref{obs:t-step-up} and~\ref{obs:t-composed} for densities of monomorphisms also hold.

\subsection{Testing permutation parameters}

A {\em permutation parameter} $f$ is a function from $\perm$ to $\mathbb{R}$. A parameter $f$ is {\em finitely forcible} if there exists a finite family of permutations $\AA$ such that for every $\varepsilon > 0$ there exist an integer $n_0$ and a real $\delta>0$ such that if $\sigma$ and $\pi$ are permutations of order at least $n_0$ satisfying $|t(\tau,\sigma)-t(\tau,\pi)| < \delta$ for every $\tau \in\AA$, then $|f(\sigma)-f(\pi)| < \varepsilon$. The set $\AA$ is referred to as a {\em forcing family for $f$}.

A permutation parameter $f$ is {\em finitely approximable} if for every $\varepsilon > 0$ there exist $\delta>0$, an integer $n_0$ and a finite family of permutations $\AA_{\varepsilon}$ such that if $\sigma$ and $\pi$ are permutations of order at least $n_0$ satisfying $|t(\tau,\sigma)-t(\tau,\pi)| < \delta$ for every $\tau \in\AA_{\varepsilon}$, then $|f(\sigma)-f(\pi)| < \varepsilon$. 

A permutation parameter $f$ is {\em testable} if for every $\varepsilon>0$ there exist an integer $n_0$ and $\tilde{f}:S_{n_0}\rightarrow\mathbb{R}$ such that for every permutation $\sigma$ of order at least $n_0$, a randomly chosen subpermutation $\pi$ of $\sigma$ of size $n_0$ satisfies $|f(\sigma)-\tilde{f}(\pi)|<\varepsilon$ with probability at least $1-\varepsilon$. The following was given in~\cite{bib-hoppen-test}.

\begin{lemma}\label{lemma:test-app}
A bounded permutation parameter $f$ is testable if and only if it is finitely approximable.
\end{lemma}

\section{Properties of the sets $T^{q}$ and $T^{q}_{\mon}$}\label{sec:Tq}

In this section, we show that densities of non-trivial indecomposable permutations are mutually independent and, more generally, that $T^{q}$ contains a ball. We start by considering the linear span of $T^{q}$.

\begin{lemma}\label{lemma:span}
For every $q\in \mathbb{N}$, ${\rm span}(T^q)=\mathbb{R}^r$, where $r$ is the number of non-trivial indecomposable permutations of order at most $q$.
\end{lemma}

\begin{proof}
Let $\{\tau_1, \ldots, \tau_r\}$ be the set of all non-trivial indecomposable permutations of order at most $q$. For a contradiction, suppose that ${\rm span}(T^q)$ has dimension less than $r$, i.e., there exist reals $c_1,\ldots, c_r$, not all of which are zero, such that
\begin{equation*}
\sum_{i=1}^r c_i v_i=0
\end{equation*}
for every $(v_1,\ldots, v_r)\in {\rm span}(T^q)$. Therefore,
\begin{equation*}
\sum_{i=1}^r c_i t(\tau_i,\Phi)=0
\end{equation*}
 for every permuton $\Phi\in \permutons$.

Consider the permutations $\tau_i$ such that $c_i\not=0$. Among these pick a $\tau_k$ of maximum order. Observation~\ref{obs:t-step-up} yields that the following holds for $s=|\tau_k|$ and every $\mathbf{x}=(x_1,\ldots, x_s)\in \mathbb{R}_+^s$ such that $\sum_{i=1}^s x_i\leq 1$:
  
\[\sum_{i=1}^r c_i t(\tau_i,\Phi^{\mathbf{x}}_{\tau_k})=
 \sum_{i=1}^r c_i |\tau_i|!\sum_{\PP\in\simplif(\tau_i)}\sum_{\psi\in \Ind(\tau_i/\PP,\tau_k)}\prod_{j=1}^{|\PP|}\frac{x_{\psi(j)}^{|P_j|}}{|P_j|!}= p(x_1, \ldots, x_s),\]
where $p$ is a polynomial. We now argue that $p$ is a polynomial of degree $s$ (and therefore it is a non-zero polynomial). Clearly, the polynomial $p$ has degree at most $s$. Since $\Ind(\tau',\tau_k)=\emptyset$ for every $\tau'$ of order $s$ such that $\tau'\neq \tau_k$, 
$c_k s! x_1x_2\cdots x_s$ is the only term of $p$ containing the
monomial $x_1x_2\cdots x_s$ with nonzero coefficient. Therefore, there exists
$\mathbf{x}$ such that $\sum_{i=1}^r c_i t(\tau_i,\Phi^{\mathbf{x}}_{\tau_k})\neq 0$, which is a contradiction. 
\end{proof}

Now, we will prove the main result of this section, Theorem~\ref{thm:ball}. It shows that the interior of $T^q$ is non-empty. Observation~\ref{obs:lin-transf} yields the same conclusion for $T_{\mon}^q$. 

\begin{proof}[Proof of Theorem~\ref{thm:ball}]
Let $\{\tau_1,\ldots, \tau_r\}$ be the set of all non-trivial indecomposable permutations of order at most $q$ and let $\Phi_1,\ldots, \Phi_r$ be permutons such that $\{\ttq(\Phi_i)
\mid i=1,\ldots, r\}$ spans $\mathbb{R}^r$. Consider the matrix $V=(v_{i,j})_{i,j=1}^{r}$, where $v_{i,j}=t(\tau_j,\Phi_i)$. Observe that the matrix $V$ is non-singular.

Consider a vector $\mathbf{x}=(x_1,\ldots,x_r)\in (0,r^{-1})^r$ and let $\Phi^{\mathbf{x}}=\bigoplus_{i\in[r]}(x_i, \Phi_i)$. 
By Observation~\ref{obs:t-composed}, we have 
\[t(\tau_j,\Phi^{\mathbf{x}})=\sum_{i=1}^r x_i^{|\tau_j|} t(\tau_j,\Phi_i)=\sum_{i=1}^t x_i^{|\tau_j|} v_{i,j}.\]

Let $\Psi$ be a map from $\mathbb{R}^r$ to $\mathbb{R}^r$ such that 
\[\Psi_j(\mathbf{x})=\sum_{i=1}^r x_i^{|\tau_j|} v_{i,j} \mbox{ for all }j\in[r].\]
Since we have $\Psi(\mathbf{x})=\ttq(\Phi^{\mathbf{x}})$, we get that
\[\Psi((0,r^{-1})^r)=\{\Psi(\mathbf{x})\mid \mathbf{x}\in (0,r^{-1})^r\}\subseteq T^{q}.\]
The Jacobian $\Jac(\Psi)(\mathbf{x})$ is a polynomial in $x_1,\ldots, x_r$. Since for $x_1=\cdots=x_r=1$ we have
\[\Jac(\Psi)=\det (v_{i,j}\cdot |\tau_j|)_{i,j=1}^{r} =\left(\prod_{j=1}^r|\tau_j|\right)\det V\neq 0,\]
$\Jac(\Psi)$ is a non-zero polynomial.

Hence, there exists $\mathbf{x}\in (0,r^{-1})^r$ for which $\Jac(\Psi)(\mathbf{x})\neq 0$. Consequently, $T^{q}$ contains a ball around $\mathbf{w}$ for $\mathbf{w}=\Psi(\mathbf{x})$.
\end{proof}

Theorem~\ref{thm:ball} implies that for every finite family $\AA$ of indecomposable permutations, there exist permutons $\Phi$ and $\Phi'$  and an indecomposable permutation $\tau$ such that $t(\pi, \Phi)= t(\pi, \Phi')$ for every $\pi\in \AA$ and $t(\tau, \Phi)\neq t(\tau, \Phi')$. The following lemma shows that an analogous statement holds for any finite family of permutations, not only for indecomposable permutations.

\begin{lemma}\label{lemma:indep}
For every finite set of permutations $\AA=\{\tau_1, \ldots, \tau_k\}$, there exists a permutation $\tau$ and permutons $\Phi$ and $\Phi'$
such that $t(\tau_i, \Phi)= t(\tau_i, \Phi')$ for every $i\in [k]$ and $t(\tau, \Phi)\neq t(\tau, \Phi')$.
\end{lemma}

\begin{proof}
Let $\BB=\{\pi_1, \ldots, \pi_{k+1}\}$ be a family of indecomposable permutations each of order $n$ with $n>|\tau_i|$ for every $i\in [k]$, such that for every $\pi_j\in \BB$, there is no $\ell<n$ satisfying $\pi_{j}(\ell+1)=\pi_{j}(\ell)+1$. We call permutations with this property {\em thorough}.
By~(\ref{eq:conn}) in Section~\ref{sec:basics} a random permutation of order $n$ is indecomposable with probability tending to one as $n$ tends to infinity. Moreover, by~(\ref{eq:simple}) in Section~\ref{sec:basics} such permutations are thorough with probability bounded away from zero, because every simple permutation is thorough. Therefore, a family $\BB$ of $k+1$ indecomposable thorough permutations exists for $n$ sufficiently large. 

Let $\Phi^{\mathbf{u}}=\bigoplus_{i\in[k+1]}(u_i, \Phi_{\pi_i}^{\mathbf{n}})$ for $\mathbf{u}=(u_1,\ldots, u_{k+1})\in (0,\frac{1}{k+1}]^{k+1}$ where  $\mathbf{n}=(\underbrace{1/n,\ldots,1/n}_{n\times})$.

Observe that for a thorough permutation $\pi$, the partition into singletons is the only $\pi$-compressive partition. Hence, by Observations~\ref{obs:t-step-up} and~\ref{obs:t-composed}, $t(\pi_i,\Phi^{\mathbf{u}})=n!(u_i/n)^n$ for every $i\in [k+1]$.
For every $j\in [k]$, the function $\mathbf{u}\mapsto t(\tau_j,\Phi^{\mathbf{u}})$ is continuous for every $j\in [k]$. We consider the continuous map $\Gamma$ from $(0,1/(k+1)]^{k+1}$ to $\mathbb{R}^k$ such that
\[\Gamma(\mathbf{u})=(t(\tau_1,\Phi^{\mathbf{u}}),\ldots,t(\tau_k,\Phi^{\mathbf{u}})).\]

Now, consider any $k$-dimensional sphere in $(0,1/(k+1)]^{k+1}$. The Borsuk-Ulam Theorem~\cite{bib-borsuk} yields the existence of two distinct points on its surface that are mapped by $\Gamma$ to the same point in $[0,1]^k$. Hence, there exist distinct $\mathbf{v}=(v_1,\ldots, v_{k+1})$ and $\mathbf{v}'=(v'_1,\ldots, v'_{k+1})$ such that $t(\tau_j,\Phi^{\mathbf{v}})=t(\tau_j,\Phi^{\mathbf{v}'})$ for every $j\in [k]$. However, if, say $v_i \neq v'_i$, then $t(\pi_i,\Phi^{\mathbf{v}})=n!(v_i/n)^{n}\neq n!(v'_i/n)^{n} = t(\pi_i,\Phi^{\mathbf{v}'})$. Therefore, we may take $\tau=\pi_i$, $\Phi=\Phi^{\mathbf{v}}$, and $\Phi'=\Phi^{\mathbf{v}'}$.
\end{proof}

\section{Non-forcible approximable parameter} 

For this section, we fix a sequence $(\tau_i)_{i\in\mathbb{N}}$ of
permutations of strictly increasing orders that satisfies the following:  For every $k>1$, there exist permutons $\Phi_k$ and $\Phi'_k$ such that $t(\sigma, \Phi_k)= t(\sigma, \Phi'_k)$ for every permutation $\sigma$ of order at most $|\tau_{k-1}|$, and $t(\tau_k, \Phi_k)>t(\tau_k, \Phi'_k)$. Such a sequence $(\tau_i)_{i\in\mathbb{N}}$ exists by Lemma~\ref{lemma:indep}. We fix such $\Phi_k$ and $\Phi'_k$ for all $k\in\mathbb{N}$ for the rest of this section. Let $\gamma_k=t(\tau_k, \Phi_k)-t(\tau_k, \Phi'_k)$ for every $k\in\mathbb{N}$.

Let $(\alpha_i)_{i\in \mathbb{N}}$ be a sequence of positive reals satisfying $\sum_{i\in \mathbb{N}} \alpha_i <1/2$ and $\sum_{i>k} \alpha_i < \alpha_k\gamma_k/4$ for every $k$.
The main result of this section is that the permutation parameter 
\[\ff(\sigma)=\sum_{i\in \mathbb{N}} \alpha_i t(\tau_i,\sigma)\]
is finitely approximable but not finitely forcible.

\begin{lemma}\label{lemma:approx}
 The permutation parameter $\ff$ is finitely approximable.
\end{lemma}

\begin{proof}
Let $\varepsilon>0$ be given. Since the sum $\sum_{i\in \mathbb{N}} \alpha_i$ converges, there exists $k$ such that $\sum_{i> k} \alpha_i <\varepsilon/2$.
Set $\AA=\{\tau_1, \ldots, \tau_{k}\}$ and $\delta=\varepsilon$.
Consider two permutations $\sigma$ and $\pi$ that satisfy $|t(\tau,\sigma)-t(\tau,\pi)|<\delta$ for every $\tau\in \AA$.
We obtain that
\begin{align*}
|\ff(\sigma)-\ff(\pi)|&= \left|\sum_{i\in \mathbb{N}} \alpha_i (t(\tau_i,\sigma)-t(\tau_i, \pi))\right|\\
&\leq \sum_{i\in \mathbb{N}} \alpha_i \left|t(\tau_i,\sigma)-t(\tau_i, \pi)\right|\\
&< \sum_{i\leq k} \alpha_i \delta + \sum_{i> k} \alpha_i|t(\tau_i,\sigma)-t(\tau_i, \pi)|\\
&< \delta/2 + \sum_{i> k} \alpha_i \cdot 1< \varepsilon.
\end{align*}
It follows that the parameter $\ff$ is finitely approximable.
\end{proof}

In the following lemma, we show that $\ff$ is not finitely forcible.

\begin{lemma}\label{lemma:non-forc}
 The permutation parameter $\ff$ is not finitely forcible.
\end{lemma}

\begin{proof}
Suppose that $\ff$ is finitely forcible and that $\AA$ is a forcing family for $\ff$. Let $\tau_i,\gamma_i, \Phi_i$ and $\Phi'_i$ be as in the definition of $\ff$ and let $k$ be such that maximum order of a permutation in $\AA$ is at most $|\tau_{k-1}|$.
We have $t(\rho, \Phi_k)=t(\rho, \Phi'_k)$ for every $\rho\in \AA$, $t(\tau_i, \Phi_k)=t(\tau_i, \Phi'_k)$ for every $i<k$, and $t(\tau_{k}, \Phi_k)- t(\tau_{k}, \Phi'_k)=\gamma_{k}$. 

Let $\varepsilon=\alpha_{k}\gamma_k/4$. Let $\delta>0$ be as in the definition of finite
forcibility of $\ff$. Without loss of generality we may assume that $\delta<\varepsilon$. 

There exist a $\Phi_k$-random permutation $\sigma$ and a $\Phi'_k$-random permutation $\sigma'$ such that $|t(\rho,\sigma)-t(\rho,\sigma')| < \delta$ for every $\rho \in\AA$, $|t(\tau_i,\sigma)-t(\tau_i,\sigma')| < \delta$ for every $i<k$ and $t(\tau_{k}, \sigma)- t(\tau_{k}, \sigma')>\gamma_{k}-\delta>3\gamma_{k}/4$.
Let us estimate the sum in the definition of $\ff$ with the $k$-th term missing. 
\begin{align*}
\left|\sum_{i\in \mathbb{N},i\neq k} \alpha_i (t(\tau_i,\sigma)-t(\tau_i, \sigma'))\right|\\
&\hspace{-6.5em}=\left|\sum_{i< k} \alpha_i \left(t(\tau_i,\sigma)-t(\tau_i, \sigma')\right)+\sum_{i>k} \alpha_i \left(t(\tau_i,\sigma)-t(\tau_i, \sigma')\right)\right|\\
&\hspace{-6.5em}<\sum_{i<k}\alpha_i\delta + \sum_{i>k}\alpha_i<\frac{\alpha_k\gamma_{k}}{8}+\frac{\alpha_k\gamma_{k}}{4}<\frac{\alpha_k\gamma_{k}}{2}
\end{align*}
This leads to the following 
\begin{align*}
|\ff(\sigma)-\ff(\sigma')|&= \left|\sum_{i\in \mathbb{N}} \alpha_i \left(t(\tau_i,\sigma)-t(\tau_i, \sigma')\right)\right|\\
&\geq  \alpha_k \left(t(\tau_{k}, \sigma)- t(\tau_{k}, \sigma')\right) - \left|\sum_{i\in \mathbb{N},i\neq k} \alpha_i \left(t(\tau_i,\sigma)-t(\tau_i, \sigma')\right)\right|\\
&> \frac{3}{4} \alpha_{k}\gamma_{k}-\frac{\alpha_k\gamma_{k}}{2}=\frac{\alpha_k\gamma_{k}}{4}=\varepsilon.
\end{align*}
This contradicts our assumption that $\ff$ is finitely forcible.
\end{proof}

Lemmas~\ref{lemma:approx} and~\ref{lemma:non-forc} imply Theorem~\ref{thm:main}. Recall that, by Lemma~\ref{lemma:test-app} the testable bounded permutation parameters are precisely the finitely approximable ones.

\section*{Acknowledgements}

This work has received funding from the European Research Council (ERC) under the European Union’s Horizon 2020 research and innovation programme (grant agreement No 648509) and under the European Union's Seventh Framework Programme (FP7/2007-2013)/ERC grant agreement no.~259385. Glebov also acknowledges the support of the ERC grant ``High-dimensional combinatorics''. Hoppen acknowledges the support of FAPERGS~(Proc.\,2233-2551/14-0), CNPq (Proc.~448754/2014-2 and~308539/2015-0) and FAPESP (Proc. 2013/03447-6). Hoppen and Kohayakawa acknowledge the support of the University of S\~ao Paulo, through NUMEC/USP (Project MaCLinC/USP).  Kohayakawa was partially supported by FAPESP (2013/03447-6,
2013/07699-0), CNPq (459335/2014-6, 310974/2013-5) and the NSF (DMS
1102086).

This work was done while Glebov and Klimo\v{s}ov\'a were at Mathematics Institute and DIMAP, University of Warwick, Coventry CV4 7AL, UK and Liu was at Department of Mathematical Sciences, University of Illinois at Urbana-Champaign, Urbana, Illinois 61801, USA.

\bibliography{biblio}

\end{document}